\documentclass[a4paper,UKenglish,cleveref, autoref, thm-restate]{lipics-v2021}
\usepackage{xspace}
\usepackage{pgfplots}

\graphicspath{{./graphics/}}%helpful if your graphic files are in another directory

\title{Depth first representations of \texorpdfstring{$k^2$}~-trees} %TODO Please add

\author{Gabriel Carmona\footnote{Corresponding author.}}{Department of Computer Science, University of Pisa, Italy}{gabriel.carmona@phd.unipi.it}{https://orcid.org/0009-0004-1454-2940}{}%TODO mandatory, please use full name; only 1 author per \author macro; first two parameters are mandatory, other parameters can be empty. Please provide at least the name of the affiliation and the country. The full address is optional. Use additional curly braces to indicate the correct name splitting when the last name consists of multiple name parts.

\author{Giovanni Manzini}{Department of Computer Science, University of Pisa, Italy}{giovanni.manzini@unipi.it}{https://orcid.org/0000-0002-5047-0196}{}

\authorrunning{G. Carmona and G. Manzini} %TODO mandatory. First: Use abbreviated first/middle names. Second (only in severe cases): Use first author plus 'et al.'

\Copyright{Gabriel Carmona and Giovanni Manzini} %TODO mandatory, please use full first names. LIPIcs license is "CC-BY";  http://creativecommons.org/licenses/by/3.0/

\ccsdesc[500]{Theory of computation~Data structures design and analysis}

\keywords{Web graphs, Sparse binary matrices, Succinct tree representations, Compact data structures} %TODO mandatory; please add comma-separated list of keywords

\category{} %optional, e.g. invited paper

\relatedversion{} %optional, e.g. full version hosted on arXiv, HAL, or other respository/website
\acknowledgements{The authors whish to thank Paolo Ferragina and Francesco Tosoni for the many fruitful discussions on $k^2$-tree representations.}%optional

\nolinenumbers %uncomment to disable line numbering

\EventEditors{John Q. Open and Joan R. Access}
\EventNoEds{2}
\EventLongTitle{42nd Conference on Very Important Topics (CVIT 2016)}
\EventShortTitle{CVIT 2016}
\EventAcronym{CVIT}
\EventYear{2016}
\EventDate{December 24--27, 2016}
\EventLocation{Little Whinging, United Kingdom}
\EventLogo{}
\SeriesVolume{42}
\ArticleNo{23}
\newcommand{\ignore}[1]{}
\newcommand{\GCremark}[1]{\marginpar{\tiny \flushleft{[GC]~#1}}}

\newcommand\skipvs{skip values\xspace}

\newcommand\identical{interchangeable\xspace}

\newcommand{\rank}{\mbox{\rm\texttt{rank}}}
\newcommand{\select}{\mbox{\rm\texttt{select}}}
\newcommand{\findc}{\mbox{\rm\texttt{find\_close}}}
\newcommand{\findo}{\mbox{\rm\texttt{find\_open}}}
\newcommand{\pppp}{\mbox{\rm\texttt{(())}}}

\newcommand{\opp}{\mbox{\rm\texttt{(}}}
\newcommand{\cpp}{\mbox{\rm\texttt{)}}}
\newcommand{\BP}{B}

\newcommand{\CBP}{B_c} 
\newcommand{\Pru}{S_c}
\newcommand{\Tar}{R_c} % reference
\newcommand{\Len}{L_c}

\pgfplotsset{compat=1.18}

\begin{document}

\maketitle

\begin{abstract}
The \textit{$k^2$-tree} is a compact data structure designed to efficiently store sparse binary matrices by leveraging both sparsity and clustering of nonzero elements. This representation supports efficiently navigational operations and complex binary operations, such as matrix-matrix multiplication, while maintaining space efficiency. The standard $k^2$-tree follows a level-by-level representation, which, while effective, prevents further compression of identical subtrees and it si not cache friendly when accessing individual subtrees. In this work, we introduce some novel depth-first representations of the $k^2$-tree and propose an efficient linear-time algorithm to identify and compress identical subtrees within these structures. 
Our experimental results show that the use of a depth-first representations is a strategy worth pursuing: for the adjacency matrix of web graphs exploiting the presence of identical subtrees does improve the compression ratio, and for some matrices depth-first representations turns out to be faster than the standard $k^2$-tree in computing the matrix-matrix multiplication.

%Our experimental results demonstrate that this new depth-first approach improves matrix-matrix multiplication performance in cases where the matrix remains sparse but exhibits a moderate density of nonzero elements. Additionally, in real-world datasets, a substantial number of identical subtrees can be detected, leading to improved compression. Although this enhanced compression comes at the cost of additional computational overhead, our results suggest that the use of depth-first representations is a strategy worth pursuing. 
\end{abstract}

\section{Introduction}

The need to handle larger and larger datasets has lead to the development of {\em compressed data structures}~\cite{gonzalo} where the data is compressed in a such a way that most operations do not require decompression. 
For storing binary matrices, the \textit{$k^2$-tree} \cite{brisaboa2009k2} is a compact representation designed to save space by leveraging the sparsity of the input and the potential clustering of nonzero elements. For a matrix $M$ this representation supports  operations such as: checking $M[p, q]$, retrieving all $j$ such that $M[p, j] = 1$ (or $M[j, q] = 1$), and finding all $(i, j)$ pairs satisfying $p_1 \leq i \leq p_2$, $q_1 \leq j \leq q_2$, and $M[i, j] = 1$. 
More recently~\cite{arroyuelo2025evaluating,7149294} described efficient algorithms for binary relations and Boolean matrix multiplication.
The $k^2$-tree has demonstrated excellent performance in various domains, including Web Graphs~\cite{BRISABOA2014152, claude2010fast}, Graph Databases~\cite{arroyuelo2025evaluating, 6824442}, Geographic Information Systems~\cite{DEBERNARDO202386,ladra2017scalable}, and Social Networks~\cite{DEBERNARDO202386}.

An intriguing question is whether the $k^2$-tree can be further compressed:
in this paper we investigate the possibility of exploiting the presence of identical subtrees.
Specifically, if two subtrees are identical, compression can be achieved by retaining only one copy and replacing the others with pointers to the retained subtree. However, implementing this approach in the current level-by-level representation of the $k^2$-tree is challenging. To efficiently identify identical subtrees, an alternative depth-first representation appears too be more suitable.

For this reason, we introduce several novel depth-first representations of the $k^2$-tree and we show how to identify and compress identical subtrees in a depth-first representation in linear time. We compare four of our representations with the original $k^2$-tree by evaluating their performance on matrices derived from Web Graphs and Random Matrices. For Web Graph matrices our experiments show that by pruning identical subtrees we can improve over the compression ratio of the original $k^2$-tree, even if at the cost of additional computational overhead. In addition, we show that in some cases depth-first representations improve matrix-matrix multiplication performance, probably because such representations are more cache friendly.  
Overall, our results show that the use of depth-first representations is a strategy worth pursuing and in Section~\ref{sec:conc} we describe many possible lines of further investigation.

% The experimental results indicate that this new depth-first representation improves matrix-matrix multiplication performance in cases where the density of ones remains relatively low but still sparse. Additionally, in such real-world datasets, a significant number of identical subtrees can be detected, leading to better compression. Although this enhanced compression comes at the cost of additional computational overhead, our results suggest that the use of depth-first representations is a strategy worth pursuing. 

\section{Notation}\label{sec:notation}

Let $\Sigma$ be a finite ordered alphabet of size $\sigma$. 
A \emph{string}, (or sequence or array), of length $n$ over alphabet $\Sigma$ is denoted with $S[1,n] \in \Sigma^n$. We write $S[i..j]$ to denote the substring $S[i]S[i+1]\cdots S[j]$ if $1\leq i \leq j \leq n$ or the empty string otherwise. The Suffix Array~\cite{MM93} of $S$ is a permutation of the integers $\{1,\ldots,n\}$ such that for $i=2,\ldots,n$, $S[SA[i-1],n] \prec S[SA[i],n]$, where $\prec$ denotes the lexicographic ordering. The LCP Array~\cite{MM93} $LCP[2,n]$ is an array of integers such that  $i=2,\ldots,n$, $LCP[i]$ is the length of the longest common prefix between $S[SA[i-1],n]$ and $S[SA[i],n]$. Both $SA$ and the $LCP$ array can be computed in $O(n)$ time~\cite{KSB06,karkkainen2009permuted,KoAlu03,KSPP03}.

Given a sequence $S$ and a symbol $c\in\Sigma$,  $\rank_c(S,i)$ returns the number of occurrences of $c$ in $S[1,i]$, and $\select_c(S,i)$ returns the position of the $i$-th occurrence of $c$ in $S$ or $-1$ if such occurrence does not exist. Using additional $o(n)$ bits of space we can preprocess $S$ so that both $\rank$ and $\select$ operations can be computed in $O(1)$ time~\cite{Munro96}. %\GCremark{add time for non-constant alphabet?}

If $\Sigma = \{ \opp, \cpp \} $, the sequence $S$ contains the same number of \opp\ and  \cpp\ symbols and no prefix $S[1,i]$ contains more $\cpp$ than $\opp$ then $S$ is called a {\em balanced parenthesis} sequence. Given a balanced parenthesis sequence $S$, $\findc(S,i)$ (resp. $\findo(S,i)$) returns the index of the closing (opening) parenthesis that matches a given opening (closing) parenthesis $S[i]$. Such operations can be supported $O(1)$ time~\cite{Munro_Raman_2002}.

\section{The canonical representation of \texorpdfstring{$k^2$}~-trees}

Given a square binary matrix $M$, the corresponding $k^2$-tree~\cite{brisaboa2009k2} is built by recursively dividing $M$ using an \textit{MX-Quadtree} strategy \cite{samet2006foundations}, splitting it into $k^2$ equal-sized submatrices, each containing $n^2/k^2$ cells. Each submatrix corresponds to a child of the root node, with its value set to $1$ if at least one cell in the submatrix is $1$; otherwise, its value is $0$. The process is recursively applied to all submatrices with at least one $1$ until all cells have been processed, see Figures~\ref{fig:matrix-example} and~\ref{fig:k2-tree-example}. 
This approach works seamlessly when $ n $ is a power of $ k $. In cases where $ n $ is not a power of $ k $, the matrix can be padded with additional rows and columns of zeros until its size becomes a power of $ k $, ensuring proper tree construction. In the following we will assume $k$ is a constant, so that $k^h$ (the size of the padded matrix) is $O(n)$.

% In the following we will always assume that $n$ is a power of $k$ and in all examples that $k=2$. 

\begin{figure}[htbp]
  \centering
  \includegraphics[scale=0.70]{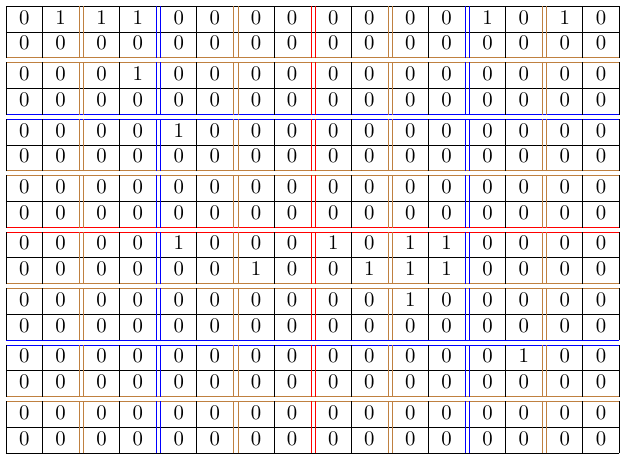}
    \caption{A $16\times16$ matrix split recursively into $2\times 2$ submatrices for the construction of a $k^2$-tree with $k = 2$.}
    \label{fig:matrix-example}
\end{figure}

\begin{figure}[tbp]
  \centering
  \includegraphics[scale=0.70]{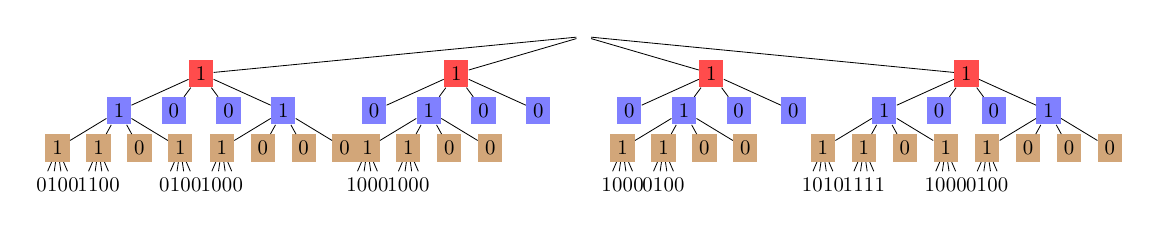}
  \caption{The $k^2$-tree, with $k=2$, representing the matrix of Figure~\ref{fig:matrix-example}. The color of each the internal nodes is related to the size of the submatrix it represents.}
  \label{fig:k2-tree-example}
\end{figure}

By construction, the $k^2$-tree is structured as a $k^2$-ary tree of height $h=\lceil \log_k n\rceil$, where each node, except the root, stores a single bit of information. For the nodes at level $\ell<h$ internal nodes store a $1$ while leaf nodes store a $0$ meaning that the corresponding $k^{h-\ell}\times k^{h-\ell}$ submatrix only contains zero elements. 
At level $h$ all nodes are leaves each one stores an actual bit value of the input matrix. 
In order to represent succinctly the input matrix, it is necessary to represent succinctly the 
above tree. In the same paper, the authors propose a representation consisting of 
two bit arrays:
\begin{itemize}
    \item $T$ (tree): stores all the bits that are not in the last level of the tree. The bits are placed following a level-wise traversal.
    \item $L$ (leaves): stores all the bits of the last level of the tree. Then, it represents some of the real original cells of the adjacency matrix.
\end{itemize}

% a bitvector $T$ containing the bits associated to the nodes at levels $\ell<h$ considered from top to bottom and left to right, and a bitvector 

For example, for the $k^2$-tree in~Figure \ref{fig:k2-tree-example}, the corresponding $T$ and $L$ bit arrays are as follows (colors are used to represent tree levels following Figure~\ref{fig:k2-tree-example}):
\begin{itemize}
    \item $T$: \texttt{{\color{red}1111} {\color{blue}1001 0100 0100 1001} {\color{brown}1101 1000 1100 1100 1101 1000}}
    \item $L$: \texttt{0100 1100 0100 1000 1000 1000 1000 0100 1010 1111 1000 0100}
\end{itemize}
The reason for which the two bit arrays $T$ and $L$ are considered separately is technical. To access specific elements of the input matrix it is necessary to efficiently navigate the $k^2$-tree. The authors observed that we can compute the $i$-th child of a node at position $x$ using the formula $\text{\texttt{child}}(x, i) = \rank_1(x)\cdot k^2+i$~\cite{brisaboa2009k2}. Hence, adding a $o(|T|)$ bits data structure supporting the \texttt{rank$_1$} operation over $T$ in constant time, it is possible to efficiently navigate the tree. 
Later, Brisaboa et al.~\cite{BRISABOA2014152} extended the functionalities of the $k^2$-tree with operations still relaying only the \rank operation on~$T$. There is no need to support the \rank operation over $L$, so it can be stored using a different (simpler) representation.

For an $n\times n$ matrix with $m$ nonzeros, it is shown that the number of bits in this representation is bounded~\cite{brisaboa2009k2} by
\begin{equation} 
k^2m\left(\log_{k^2} \frac{n^2}{m} + O(1)\right)
\end{equation}
However, such value is obtained for pathological inputs and the behavior in practice is much better. Brisaboa et al.~\cite{BRISABOA2014152} in experiments with web graphs, found that the actual space usage is significantly lower and competitive with other graph compression schemes.

\section{Depth-First Traversal Representations}

The Canonical representation is extremely space efficient since it essentially uses a single bit for each node. Considering that it supports tree navigation in constant time it seems hard to improve. However, such representation appears ill suited to exploit the possible presence of identical submatrices in the input matrix. Such submatrices corresponds to identical subtrees but with the Canonical representation the information on every subtree is partitioned into the different tree levels. It is therefore non trivial to detect equal subtrees within the Canonical representation. Another issue of the Canonical representation is that nodes which are close in the tree can be stored far apart: for some computation this layout can generate a large number cache misses. 

As a possible solution to the above shortcomings, in this section we describe alternative representations based on a depth first traversal of the $k^2$-tree: such approach stores together the information regarding a given subtree thus improving the locality of references and simplifying the task of detecting identical subtrees.

\subsection{Plain Depth-First representation}\label{subsec:PDF}

The simplest Depth-First representation consists in traversing the tree in depth-first order; when the visit reaches a non-leaf node $u$, we write the $k^2$ bits associated to $u$'s children. 
The resulting bit array $P$ has length $m k^2$, where $m$ is the number of internal nodes of the $k^2$-tree. See Figure~\ref{fig:k2tree-pdft} for an example.  
It is easy to see that $P$ is a permutation of the bits in $T \cup L$ of the canonical representation; more precisely both $P$ and $T \cup L$ can be partitioned into $m$ blocks of $k^2$ bits, and the blocks of $P$ are a permutation of the blocks of $T \cup L$.

We call this representation the Plain Depth First (PDF) representation since we do not add any additional information to speed up navigation. Note that given a pair of indices $i,j$ we are still able to determine the entry $M[i][j]$ of the matrix represented by the $k^2$-tree, but this requires a depth first traversal of the tree up to the node representing the largest non-empty submatrix containing position $i,j$. Although the access to a single entry is inefficient, some operations involving the whole matrix (such as computing the vector product $y = M x$) do require visiting the whole $k^2$-tree and this can be done with a simple left-to-right scan of the bit array $P$. Such visit is much more cache friendly than a visit of the canonical representation of the same tree

\ignore{
that we will show is one that we call Plain Depth First Traversal (PDFT), where we traverse in depth-first order the nodes of the tree and when we enter a node, we write the bits of the children of that node. We continue like that for each node in this order. For example, for the Figure \ref{fig:k2-tree-example}, we will obtain the following bit vector showed at the Figure \ref{fig:k2tree-pdft}.}

\begin{figure}[hbpt]
    \centering
    \includegraphics[scale=0.7]{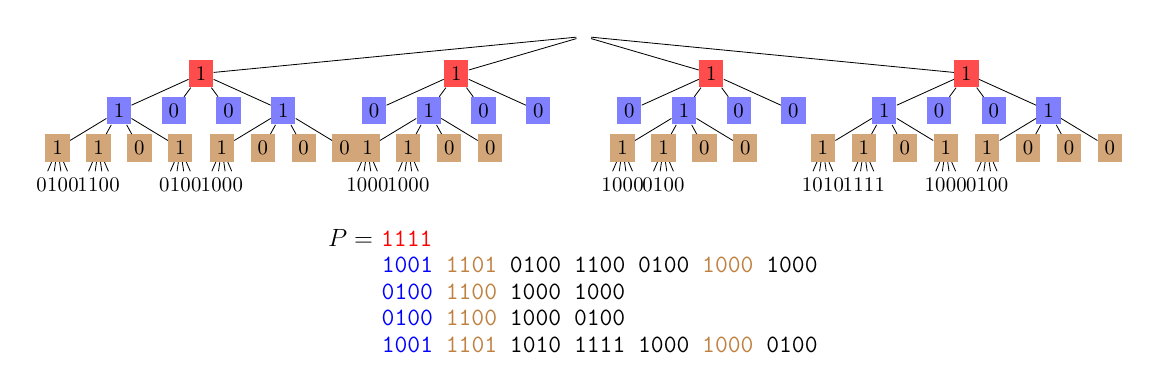}
    \caption{The Plain Depth-First representation of the same tree as Figure~\ref{fig:k2-tree-example}. The representation consists of the bit array $P$ containing the concatenation of the 4-tuples of bits shown above: the first row represents the children of the root node. Each of the subsequent four rows corresponds to the subtrees rooted at each child of the root, listed in depth-first traversal order; different colors represent different tree levels.}
    \label{fig:k2tree-pdft}
\end{figure}

% The advantage of this representation is that it uses the same number of bits as the level-by-level approach, specifically $|T| + |L|$ bits. In this case, we avoid employing a structure to support the \texttt{rank} operation. As a result, we can apply operations that require accessing the entire tree by simply scanning this sequence from left to right. However, the trade-off is that it does not support logarithmic access to individual elements.

%Despite this limitation, it proves useful due to its efficient bit usage and can be advantageous for operations that require only a full scanning of the tree or as part of a hybrid approach. For example, we could use \textit{BP} and \textit{DFUDS} for the upper levels and \textit{PDFT} for the lower levels. Therefore, in the next section, we explore Depth-First representations, which supports efficient navigation.

%{
%\texttt{{\color{red}1111}}
%
%\texttt{{\color{blue}1001} {\color{brown}1101} 0100 1100 1000 {\color{brown}1000} 1000}
%
%\texttt{{\color{blue}0100} {\color{brown}1100} 1000 1000}
%
%\texttt{{\color{blue}0100} {\color{brown}1100} 1000 1000}
%
%\texttt{{\color{blue}1001} {\color{brown}1101} 0100 1100 0100 {\color{brown}1001} 1000}}

\subsection{Enriched Depth-First representation}\label{subsec:EDF}

The main drawback of the PDF representation is that in order to find the starting position in $P$ of the subtree corresponding to the, say, third child of the root note, we must execute a visit of the first two subtrees. The cost of the visit is proportional to the number of nodes in such subtrees so this is a major problem when such subtrees are large: since the visit consists of a linear scan of a subarray of $P$, small subtrees are less of a problem. 

The above observation suggests to ``enrich'' to plain depth-first representation, with additional information that allow us to skip large subtrees without visiting them. The simplest approach is to store after each block of $k^2$ bits associated to an internal node, the size, in terms of number of blocks, of its subtrees, with the exception of the last one. In the example of Figure~\ref{fig:k2tree-pdft}, after the block {\textcolor{red}{\tt 1111}} we would encode the size of the first three subtrees, e.g. 7 (first subtree), 4 (second subtree), and 4 (third subtree). If we need to access the third subtree we skip $7+4 = 11$ blocks, while if we need to access the fourth subtree we skip $7+4+4 = 15$ blocks. Note that we do not store the size of the last subtree (7 in our example), since it is not used for skipping any subtree. We store this information only for the first $\ell$ levels of the tree: the rationale is that as we descend in the tree the number of nodes per level increases, and therefore also the number of ``skip'' values, and, as we observed above, visiting small trees is a relatively fast operation.

We can refine the above idea by choosing a threshold $\tau$ and storing the ``skip'' values only for subtrees of size larger than $\tau$. For example it we set $\tau=6$ we use ``skip'' values only for the root and for the first and last children of the root, whose subtrees have size 7. Since those children have two subtrees, they only need to store a single skip value (4 in our example). Note that using the threshold strategy we do not waste space for skip values for small subtrees which happen to be in the first levels. If $N$ is the total number of nodes and we choose $\tau=f(N)$ we can estimate the overhead of storing the \skipvs as follows. Consider the subtrees which have more than $f(N)$ nodes and do not contain any subtree with more than $f(N)$ nodes. Clearly there are at most $N/f(N)$ such subtrees, and each such subtree will have at most $\log_k n$ ancestors. Not all ancestors are distinct but we can state that the number of node containing \skipvs is at most $O(N\log n /f(N))$.  Assuming we use $O(\log N)$ bits for each skip value, the total overhead for skip values  is $O(N k^2\log N\log n /f(N))$ bits. In our experiments we set $\tau = \sqrt{N}$ and $k=4$ so the overhead is $O(\sqrt{N}\log N\log n)$ bits. %Since $N=|P|+1$, this is $o(|P|)$ when $N=\Omega(\log^2 n)$, that is when the matrix is not pathologically sparse.\GMremark{fix formula}

We call the above representation Enriched Depth First (EDF). Note that there are a few details of the representation that has to be decided during the implementation. The first one is how to encode the skip values: to save space it is preferable to use a variable length code, which one among the many existing alternatives (Elias, Rice, Variable-Byte, etc. see for example~\cite[Chapter~11]{pearls2023}) depends on the desired time-space tradeoff. The second implementation choice is {\em where} to store the ``skip'' values. A first alternative is to store them inside the bit array $P$, immediately after the $k^2$-bit block representing the root of the subtree they refer to. This alternative is cache friendly since the whole representation still consists of a single bit array which is mostly accessed in sequential scans. However, the size of the subtrees, hence the skip values, now depend also on the skip values stored at the lower levels: this makes the construction of the representation more complex since now it has to be done, bottom up and left to right (details in the full paper).

Another possible approach is to store the \skipvs in a separate array $S$. This simplifies the construction and navigation algorithms but 1) it is not cache friendly since navigation accesses two distinct bit arrays, and 2) it requires the storage of some additional information. To see this latter point, observe that in order to skip a subtree $T_i$ we need to skip the portion of the $P$ array containing the encoding of $T_i$ nodes {\em and}
the portion of the $S$ array containing the \skipvs for $T_i$. Hence the size of such portion of $S$ has to be saved as well; this can be done in the array $S$ itself, but requires some additional space.

\subsection{Balanced Parenthesis representation}\label{subsec:bp}

In this section we introduce another depth-first based representation of $k^2$-trees based on Balanced Parenthesis (\textit{BP}) encoding, following the ideas introduced by Munro and Raman~\cite{Munro_Raman_2002}. 
Our representation follows the classical approach used for general trees; however we exploit the spacial structure of the $k^2$-tree introducing an optimization for the last-level nodes. Instead of representing these nodes using parenthesis, we store their actual values in a separate bit array~$L'$.

%ordered from left to right. This bit vector efficiently supports \textit{rank} operation, enabling fast queries and compact storage.

Let $h = \lceil\log_k n\rceil$ denote the tree height. 
Formally, the construction of our BP representation is done as follows. We visit the $k^2$ tree in depth first order and, as suggested by Munro and Raman~\cite{Munro_Raman_2002}, 
we write to a vector $\BP$ a \texttt{(} every time we start the visit of a node and we write a \texttt{)} every time the visit of a node is complete and we go back to the parent node. However, when we reach a level $h-1$ node which is not a leaf, instead of visiting its children we write a pair \texttt{()} and we write the $k^2$ bits associated to its children to the bit array $L'$. See Figure~\ref{fig:k2tree-bp} for an example. 

Note that the subtrees rooted at a level $h-1$ node are represented by either \texttt{()}, if the have no children, or by \texttt{(())} if the have children.  Since in a $k^2$ tree every node has either 4 children or none, the sequence \texttt{(())} cannot represent a subtree rooted at a level $< h-1$. This implies that there is a one-to-one correspondence between the occurrences of the pattern 
\texttt{(())} in the parenthesis vector $\BP$, the subtrees with children at level $h-1$ and the blocks of $k^2$ bits in the $L'$ array. By construction, the correspondence is order preserving in the sense that the $i$-th occurrence of the pattern  \texttt{(())} corresponds to the $i$-th block in $L'$.
This means that to access the bit values stored in a level-$h$ leaf, we need to locate the position $p$ of its parent, and count the number $x$ of occurrences of the pattern  \texttt{(())} in $\BP$ up to position $p$; the desired bit values are stored starting from position $x k^2$ in $L'$.

\begin{figure}[bpt]
    \centering
    \includegraphics[scale=0.7]{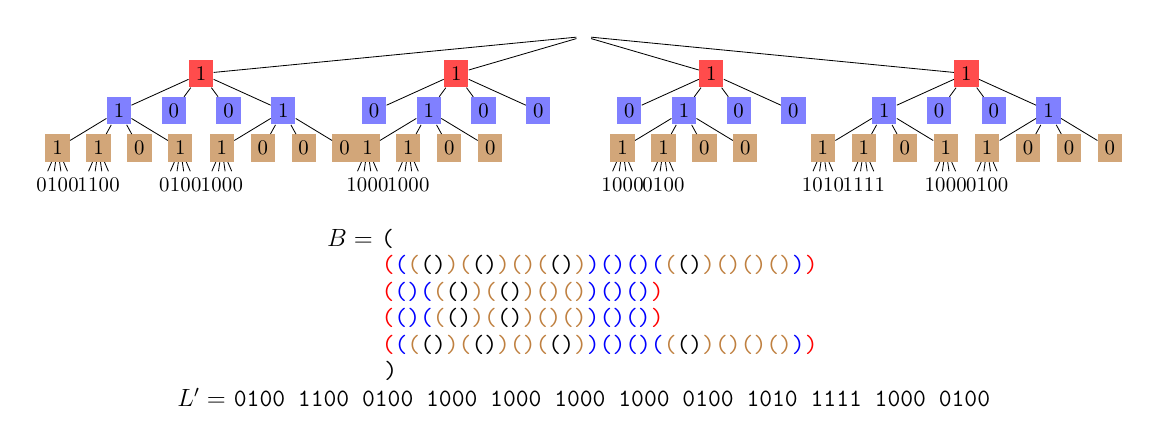}
    \caption{The balanced parenthesis representation of the $k^2$-tree of Figure \ref{fig:k2-tree-example}. The resulting parenthesis vector $\BP$ is the concatenation  of the six rows of parentheses:
    the first and last row are the open and close parenthesis for the root; and the other four rows  represents the subtrees rooted at the four children of the root. Also shown is the bit vector $L'$ containing the values in the bottom level stored according to the depth-first order visit.}
    \label{fig:k2tree-bp}
\end{figure}
%\begin{center}
%\begin{minipage}{6cm}
%\opbk
%
%\opr\opbl\opbr\opbk\clbk\clbr\opbr\opbk\clbk\clbr\opbr\clbr\opbr\opbk\clbk\clbr\clbl\opbl\clbl\opbl\clbl\opbl\opbr\opbk\clbk\clbr\opbr\clbr\opbr\clbr\opbr\clbr\clbl\clr
%
%\opr\opbl\clbl\opbl\opbr\opbk\clbk\clbr\opbr\opbk\clbk\clbr\opbr\clbr\opbr\clbr\clbl\opbl\clbl\opbl\clbl\clr
%
%\opr\opbl\clbl\opbl\opbr\opbk\clbk\clbr\opbr\opbk\clbk\clbr\opbr\clbr\opbr\clbr\clbl\opbl\clbl\opbl\clbl\clr
%
%\opr\opbl\opbr\opbk\clbk\clbr\opbr\opbk\clbk\clbr\opbr\clbr\opbr\opbk\clbk\clbr\clbl\opbl\clbl\opbl\clbl\opbl\opbr\opbk\clbk\clbr\opbr\clbr\opbr\clbr\opbr\clbr\clbl\clr
%
%\clbk
%\end{minipage}
%   
%\end{center}

Let $t$ denote the number of nodes in levels $0,\ldots h-1$ and $\ell$ denotes the number of level-$h$ leaves. It is immediate to see that the parenthesis vector $\BP$ and the bit vector $L'$ have length
\begin{equation}\label{eq:bp}
|\BP| = 2t +\frac{2\ell}{k^2},\qquad
|L'| = \ell.     
\end{equation}
The BP representation therefore takes roughly twice the space used by the Canonical and (Enriched) Depth-First representations that take $t+\ell$ bits plus possibly lower order terms to support fast navigation. The advantages of the BP representation are the locality of reference and the possibility, explored in Section~\ref{sec:subtrees} to compress identical subtrees. 

To ensure constant time navigation with the BP representation we need to support the \findc\ operation in constant time that, as recalled in Section~\ref{sec:notation}, can be supported using $o(|\BP|)$ bits of auxiliary space.  When we reach an internal node at level $h-1$, to reach the corresponding bits in the $L'$ array we need to support the constant time rank operation for the pattern $\texttt{(())}$ on the array $\BP$; this can be achieved by a straightforward modification of the rank data structures for the $\rank$ data structure still using $o(|\BP|)$ bits of auxiliary space.

% can employ the \textit{range min-max tree} by Sadakane and Navarro~\cite{navarro2014fully}, on the parenthesis array $\BP$. For any constant $c>0$ this structure supports all the required operation in constant $O(c)$ time using $O(|\BP|/\log^c(|\BP|))$ bits of space. 

%Note however that if we only need \ldots
%In our case, since we only need to support the operation of finding the $i$-th child, we can replace the \textit{range min-max tree} with an alternative data structure that utilizes auxiliary space of $o(n)$ while supporting the $i$-th child operation in constant time.\GMremark{which one?}

\subsection{Additional depth-first representations}

Another depth-first representation of $k^2$-trees can be obtained using the 
Depth First Unary Degree Sequence ({DFUDS}) representation~\cite{BDM3R}. The {DFUDS} representation builds a balanced parenthesis sequence by traversing the tree in depth-first order and writing the number of children of each node in unary using the symbols \opp\ and \cpp, so a node with two children is represented by \opp\opp\cpp\ 
(a ``super-root'' initial \opp\ symbol is needed to make the sequence balanced). As described, the above representation takes the same space~\eqref{eq:bp} as the BP representation, which is roughly twice the space of the Canonical representation. However, since in a $k^2$ tree all internal nodes have exactly four children, we can use the {\em ultra-succinct} representation by Jansson et al. \cite{jansson2012} to achieve the {\em tree degree entropy} of the $k^2$-tree, which equals one bit per node, as in Canonical representation. However, since the Jansson et al.'s representation is rather complex, we leave the study of this approach to the full paper. 

Finally, we point out that it is possible to combine the BP representation with the EDF or PDF representations. For example, we can fix a threshold $\tau$ and use the BP representation for all subtrees of size at least $\tau$. If an internal node is the root of a subtree of size smaller than $\tau$ we can represent it in $\BP$ with the special sequence $\texttt{(())}$, and store its EDF or PDF representation in an auxiliary array $L'$. We leave the investigation of this mixed approach to the full paper.

\section{Compressing subtrees in the BP representation}\label{sec:subtrees}

Identical submatrices in the original input matrix $M$ can lead to identical subtrees in the $k^2$-tree representing $M$. We now show how to detect identical subtrees and replace them with a ``pointer'' to a previous occurrence, and how to perform navigation operations on this compressed representation. The algorithm for detecting identical subtrees are based on the Suffix Array and LCP array and takes linear time, i.e. time proportional to the size of the balanced parenthesis array~$\BP$. 
We consider the BP representation since it is the one more challenging from the algorithmic point of view. In the full paper we will discuss how to find and represent identical subtrees using the PDF and EDF representation. 

When working with the BP representation, instead of considering identical subtrees in the traditional sense, we will consider the wider concept of \identical subtrees according to the following definition, that, for our convenience, excludes the pathological case of subtrees consisting of a single node.

\ignore{
\subsection{Find and replace identical subtrees}

In the case of a $k^2$-tree, it's common to encounter identical subtrees. Relying solely on the succinct representation of the tree, such as using \textit{BP}, may lead to inefficiencies in space usage due to the presence of these identical subtrees. This is because, \textit{BP}, identical subtrees will correspond to identical substrings. Therefore, detecting and compressing these repeated patterns becomes a crucial optimization problem. For instance, in Figure \ref{fig:k2-tree-example}, we can observe several identical subtrees, some of which are highlighted in Figure \ref{fig:ex-idem-subtree}. In this case, we will not consider leaves as valid subtrees.}

\begin{definition} \label{def:isubtrees}
    Let $T_1$ and $T_2$ be subtrees of height at least 1 of the $k^2$-tree $T$. We say that $T_1$ and $T_2$ are {\em \identical} if their respective balanced parenthesis sequences generated during the visit described in Section~\ref{subsec:bp} are identical.     
\end{definition}

We point out that the balanced parenthesis sequences generated during the visit in Section~\ref{subsec:bp} does not include any information regarding the last level leaves since they are stored in binary form in the vector $L'$. Hence, some subtrees are \identical even if some of the values stored in the last level differ, see for example the subtrees surrounded in red in Figure~\ref{fig:ex-idem-subtree}.

\begin{figure}[t]
    \centering
    \includegraphics[scale=0.7]{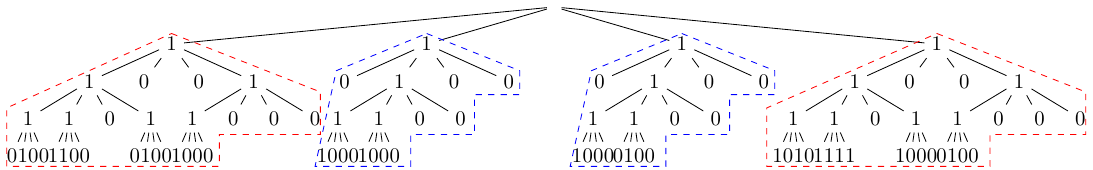}
    \caption{The subtrees surrounded by dashed lines of the same color are \identical according to Definition~\ref{def:isubtrees}. Note that the subtrees surrounded in red are \identical even if the values stored in some of the leaves at the last level are different.}
    \label{fig:ex-idem-subtree}
\end{figure}

\ignore{

In this case, we want to find identical subtrees in the $k^2$-tree. But thanks to the particular properties of the $k^2$-tree, the following theorem emerges.

\begin{theorem}\label{theo:identical}
    If any two subtrees are identical, their roots must be at the same level, and they must each have at least one node at the last level of the $k^2$-tree, except in the case where they are leaves.
\end{theorem}
\begin{proof}
 Let $T_1$ and $T_2$ be two identical subtrees, each with more than one node. Having more than one node implies that there is at least one $1$ in the $k^2$ structure, which ensures the presence of at least one node at the last level of the tree. For the trees to be identical, not only must they each have at least one node at the last level, but their roots must also be at the same level. If the roots are at different levels, the subtrees will have different heights, and consequently, they cannot be identical.   
\end{proof}
}

\begin{lemma} \label{lemma:indentical}
    The subtree starting at position $i$ in $\BP$ is \identical with the subtree starting in position $j$, if and only if, setting $c=\mbox{\rm \texttt{find\_close}}(i)$, the substring $\BP[i..c]$ is equal to $\BP[j..j+c-i]$. 
\end{lemma}

\begin{proof}
Note that by construction it is $\BP[i]=\BP[j]=\texttt{(}$. The position $c=\texttt{find\_close}(i)$ is the position of the closed parenthesis matching $\BP[i]$. Hence the visit of the entire subtree starting at position $i$ generates the string $\BP[i..c]$. If such string is identical to $\BP[j..j+c-i]$ then the two subtrees are interchangeable. 
\end{proof}

Because of the above lemma, we can detect all \identical subtrees in linear time as follows. First we compute the suffix array $SA$  and $LCP$ array of the parenthesis sequence $\BP$. For $i=1,\ldots |\BP|-1$, if $\BP[SA[i]] = \texttt{(}$, we set $c=\findc(SA[i])$. Then, if $c-SA[i] \leq  LCP[i]$ by the above lemma the subtree starting at position $SA[i]$ is \identical with the subtree starting at $SA[i-1]$. During this procedure we ignore positions $i$ when $LCP[i] \leq 4$ (or some other larger threshold) since in that case an \identical subtree would be so small that replacing it with a pointer would offer no benefit.

The above algorithm will find {\em all} \identical subtrees but for our purposes this is too much. The problem is that subtrees of \identical subtrees are still \identical but we are only interested in finding {\em maximal} \identical subtrees, i.e. subtrees which are not contained in larger \identical subtrees. Therefore, instead of scanning the $SA$ array we scan the sequence $\BP$: for $j=1,\ldots,|\BP|$ , if $\BP[j] = \texttt{(}$ we set $i=SA^{-1}[j]$ and we check whether $LCP[i] > \texttt{find\_close}(j)-j$. If this is the case and the subtree starting at $j$ is \identical with a previous subtree, we restart the scanning of $\BP$ at position $\texttt{find\_close}(j)+1$ therefore skipping all non-maximal \identical subtrees contained in the subtree starting at $j$.   
Note that if we have $k$ consecutive suffixes $SA[i],\ldots,SA[i+k-1]$ each representing \identical subtrees, then they are all \identical among themselves. In this case we 
select the subtree with the leftmost starting position as the reference subtree and the other subtrees will point to this reference subtree.

Having identified which (maximal) subtrees should be replaced by a pointer, in the following we call them {\em pruned} subtrees, the construction of the Compressed Balanced Parenthesis (CBP) representation is done as follows. 
Let $h = \lceil\log_{k} n\rceil$ denote the number of tree levels, $\ell$ denote the number of level $h$ leaves (whose value are stored in $L'$), and $p$ denote the number of pruned subtrees.
Given the uncompressed sequence of balanced parenthesis $\BP$ we build the compressed sequence $\CBP$ by replacing the subsequence representing each pruned subtree with the pattern \texttt{(())} (note $\CBP$ is still balanced). Since this is the same pattern used to represent internal nodes at level $h-1$, we also introduce a binary vector $\Pru[1,\ell + p]$ such that for each occurrence of the pattern  \texttt{(())} in $\CBP$ we store a 0 in $\Pru$ if it represent an internal node level $h-1$, and we store a 1 if it represents a pruned subtree. 

We also use an integer array $\Tar[1,p]$ such that $\Tar[v]$ is the starting position in $\CBP$ of the reference subtree for the $v$-th pruned subtree. When the navigation reaches at position $i$ of $\CBP$ a \pppp\ pattern which is not at level $h-1$, we compute $r=\rank_{\scriptsize{\pppp}}(\CBP,i)$,  $v=\rank_1(\Pru,r)$. By construction, the pruned subtree starting at position $i$ is the $v$-th pruned subtree, hence $j=\Tar[v]$ is the starting position in $\CBP$ of the reference subtree for the one starting at~$i$.
Finally, since \identical subtrees can have different values stored at level $h$ (i.e. in the bitarray $L'$), when jumping from position $i$ to position $j$ we need to compute the length of the portion of $L'$ corresponding to the tree nodes between position $j$ and position $i$. Such  lengths can be stored in another integer array $\Len[1,p]$ or computed on the fly using an auxiliary binary array containing the unary encoding of the number of level-$h$ for each pruned subtree (details in the full paper).
Figure~\ref{fig:ex-idem-subtree} shows the $k^2$-tree of Figure~\ref{fig:k2tree-bp} with two subtrees pruned and the resulting $\CBP$ sequence.

\begin{figure}[htbp]
    \centering
    \includegraphics[scale=0.7]{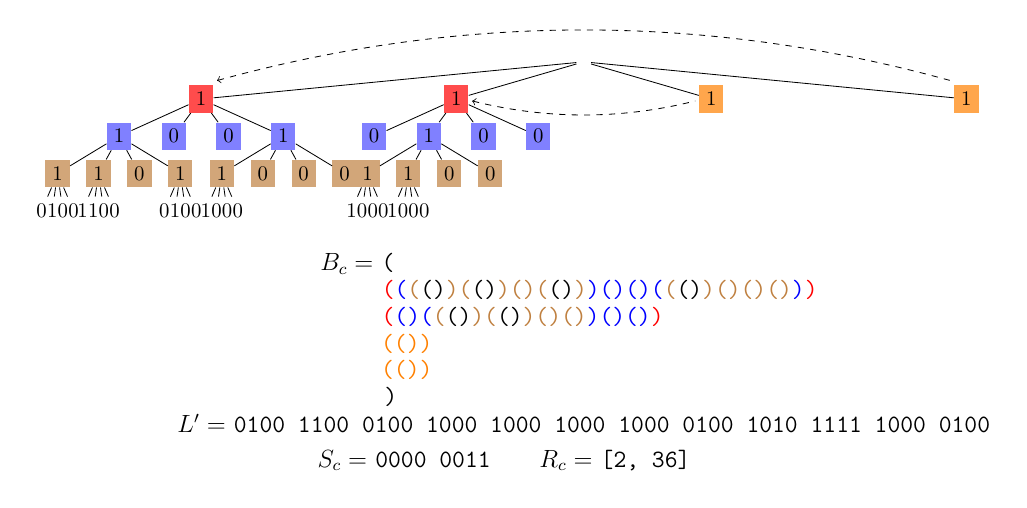}
%
%\begin{center}
%\begin{minipage}{6cm}
%\large
%\opbk
%
%\opr\opbl\opbr\opbk\clbk\clbr\opbr\opbk\clbk\clbr\opbr\clbr\opbr\opbk\clbk\clbr\clbl\opbl\clbl\opbl\clbl\opbl\opbr\opbk\clbk\clbr\opbr\clbr\opbr\clbr\opbr\clbr\clbl\clr
%
%\opr\opbl\clbl\opbl\opbr\opbk\clbk\clbr\opbr\opbk\clbk\clbr\opbr\clbr\opbr\clbr\clbl\opbl\clbl\opbl\clbl\clr
%
%\opo\opo\clo\clo
%
%\opo\opo\clo\clo
%
%\clbk
%\end{minipage}
%\end{center}

\caption{The Compressed Balanced Parenthesis representation for the $k^2$-tree shown in Figure~\ref{fig:k2tree-bp} with two subtrees (in orange) pruned and replaced with pointers to \identical subtrees. The two pointers are explicitly stored in the array $\Tar$}\label{fig:repl-k2tree}
\end{figure}

\section{Experimental Results}

In this section we report some preliminary results on our depth-first based representations of $k^2$-trees. To evaluate the effectiveness of a representation we need to assess not only its compression performance but also its efficiency in supporting matrix operations. To this end we tested our representations on the matrix-matrix multiplication operation for which there exists an efficient implementation for the canonical $k^2$ representation provided by Arroyuelo et al.~\cite{arroyuelo2025evaluating}. Such implementation is based on the classical divide-and-conquer recursive procedure: Given two matrices $A$ and $B$, if we logically partition them into four submatrices of equal size
\[
A = \begin{pmatrix}
A_0 & A_1\\
A_2 & A_3
\end{pmatrix}, \quad
B = \begin{pmatrix}
B_0 & B_1\\
B_2 & B_3
\end{pmatrix},
\]
then the product $A\times B$ can be obtained by the recursive formula
\[
A \times B = \begin{pmatrix}
A_0 \times B_0 + A_1 \times B_2 & A_0 \times B_1 + A_1 \times B_3 \\
A_2 \times B_0 + A_3 \times B_2 & A_2 \times B_1 + A_3 \times B_3
\end{pmatrix}.
\]
The above algorithm is particularly suitable when the input matrices are represented by $k^2$-tree when $k=2$, since the submatrices $A_i$'s and $B_i$'s correspond to the four children of the root node. 
Another advantage of the $k^2$-tree representation is its ability to efficiently handle zero submatrices. If any submatrix is entirely composed of zeros, its product with another submatrix will also be a zero matrix, eliminating unnecessary computations. Additionally, when summing a product $A_i \times B_j$ with a zero matrix, the operation can be optimized by directly referencing the nonzero product, avoiding redundant memory operations.

The recursive structure and the presence of multiple sums and products make this computation a good benchmark for the different $k^2$-tree implementations. Note that in all implementations the product matrix is returned in compressed form, so the multiplication algorithm also measure the efficiency of the construction algorithm.

We run our experiments on an Ubuntu machine with 394 GiB of internal memory and an Intel(R) Xeon(R) Gold 6132 CPU @ 2.60GHz.
We compare the following implementations:

\begin{itemize}
    \item \texttt{K2-TREE}\footnote{\href{https://github.com/adriangbrandon/rpq-matrix/tree/main}{\texttt{K2-TREE}: https://github.com/adriangbrandon/rpq-matrix}} this is the canonical level-by-level $k^2$ implementation by Arroyuelo et al. \cite{arroyuelo2025evaluating}.
    \item \texttt{K2-PDF} and \texttt{K2-EDF}\footnote{\href{ https://github.com/acubeLab/k2tree/tree/main}{\texttt{K2-PDF} \& \texttt{K2-EDF}: https://github.com/acubeLab/k2tree/}} implementations based on the Plain Depth First and Enriched Depth First representations described in Section~\ref{subsec:PDF} and~\ref{subsec:EDF}
    respectively. For the enriched representation we used as threshold the square root of the number of nodes. 
    \item \texttt{K2-BP} and \texttt{K2-CBP}\footnote{\href{https://github.com/Yhatoh/k2tree/tree/better_version}{\texttt{K2-BP} \& \texttt{K2-CBP}: https://github.com/Yhatoh/k2tree}} implementations based on the Balanced Parenthesis representation and the compressed Balanced Parenthesis representation described in Section~\ref{subsec:bp} and~\ref{sec:subtrees}. Both representations are implemented using the sdsl-lite library~\cite{gbmp2014sea}.
\end{itemize}

\noindent
We tested all implementations using two different datasets:
\begin{itemize}  
    \item \textbf{WebGraph} \cite{BRSLLP, BoVWFI}: The compression of the adjacency matrix of web graphs was one of the motivations for introducing $k^2$-trees. We tested four large web graphs: IN-2004, INDOCHINA-2004 (ID-2004), ARABIC-2005 (AR-2005), and UK-2005.  Information about each matrix is provided in Table~\ref{tab:info-webgraph}. Since matrices have different sizes, for this dataset matrix multiplication was performed by squaring each matrix.
    \item \textbf{Random Matrices}: Following the approach of Arroyuelo et al. \cite{arroyuelo2025evaluating}, we generated matrices with varying densities to assess the performance of matrix-matrix multiplication by measuring computation time. Specifically, we tested matrices of size $10^3 \times 10^3$ with densities $2 \times 10^{-1}$, $10^{-1}$, $10^{-2}$, $10^{-3}$, and $10^{-4}$. For each density, we generated 10 matrices with uniformly distributed values. We then performed matrix-matrix multiplication on each pair and computed the average computation time. % Due that these are random matrices, we will not take account the compression.
\end{itemize}

\begin{table}[htbp]
    \begin{center}
    \begin{tabular}{|l|r|r|r|r|}
       \hline
       Dataset & Matrix size & \# Nonzero & Nodes / \# Nonzero & Density\\
       \hline
       IN-2004 & 1,382,908 & 16,917,053 & 1.41 & $8\cdot 10^{-6}$\\
       \hline
       ID-2004 & 7,414,866 & 194,109,311 & 1.06 & $3\cdot 10^{-6}$\\
       \hline
       AR-2005 & 22,744,080 & 639,999,458 & 1.31 & $1\cdot 10^{-6}$\\
       \hline
       UK-2005 & 39,459,925 & 936,364,282 & 1.34 & $6\cdot 10^{-7}$\\
       \hline
    \end{tabular}        
    \end{center}
    \caption{Information on size, number of nonzero, density of nonzero elements for the matrices in the Web Graph dataset.}
    \label{tab:info-webgraph}
\end{table}

We first report the results for the Web Graph dataset. As shown in Table~\ref{tab:webgraph_result} the best compression is achieved by \texttt{K2-CBP} with the exception of IN-2004 where the best compression is obtained by \texttt{K2-PDF} and \texttt{K2-EDF}; this shows that web graph matrices indeed contain a large number of repeated subtrees. Comparing \texttt{K2-EDF} with \texttt{K2-PDF} shows that the space used by skip values is relatively modest, and both algorithms take less space than \texttt{K2-TREE}. 
Notice that \texttt{K2-BP} is, as expected, not space efficient and exploiting repeated subtrees reduced the overall size by a factor three. This suggests that removing repeated subtrees in the PDF and EDF representations can lead to similar significant saving.

Looking at the running times we notice that, thanks to the skip values, \texttt{K2-EDF} is always faster than \texttt{K2-PDF} by a factor roughly two. The algorithm \texttt{K2-CBP} is the slowest, probably because of the overhead of locating the reference for pruned subtrees. We plan to improve the times/space tradeoff by pruning only subtrees larger than an appropriate threshold. The comparison  with  \texttt{K2-TREE} is less clear cut. For ID-2004 \texttt{K2-EDF} is much faster than \texttt{K2-TREE} but for the other matrices \texttt{K2-TREE} is faster, significantly so for UK-2005. We conjecture this can be due to the fact that UK-2005 is significantly less dense than the other matrices. UK-2005 also has large submatrices, and \texttt{K2-TREE} appears to have a more efficient strategy for partitioning a matrix and detecting submatrices which are not to be considered since they are going to be multiplied by zero submatrices. More experiments and analysis are clearly required for a complete understanding.

\begin{table}
    \begin{center}
    \begin{tabular}{|l||c|c||c|c||c|c||c|c||c|c|}
    \hline
       \multirow{2}{*}{Dataset} & \multicolumn{2}{c||}{\texttt{K2-TREE}} & \multicolumn{2}{c||}{\texttt{K2-PDF}} &  \multicolumn{2}{c||}{\texttt{K2-EDF}} &\multicolumn{2}{c||}{\texttt{K2-BP}} & \multicolumn{2}{c|}{\texttt{K2-CBP}}\\
    \cline{2-11}
     & bits & min & bits & min & bits & min & bits & min & bits & min \\
     \hline
     IN-2004 & 3.12 & 1.07 & 2.93 & 2.20 & 2.97 & 1.26 & 7.80 & 2.61 & 2.90 & 4.40\\ 
     \hline
     ID-2004 & 2.56 & 125 & 2.41 & 54 & 2.42 & 34 & 6.20 & 124 & 1.87 & 164\\ 
     \hline
     AR-2005 & 2.92 & 85 & 2.75 & 291 & 2.75 & 134 & 7.24 & 254 & 2.37 & 335\\ 
     \hline
     UK-2005 & 2.89 & 89 & 2.72 & 901 & 2.73 & 429 & 7.16 & 425 & 2.35 & 909\\ 
     \hline
    \end{tabular}
    \end{center}
    \caption{Results for matrix multiplication for web graph matrices; table reports the compressed matrix size in bits per nonzero and running time in minutes.}
    \label{tab:webgraph_result}
\end{table}

%\begin{table}
%    \begin{center}
%    \begin{tabular}{|l||c|c||c|c||c|c||c|c||c|c||c|c||c|c|}
%    \hline
%       \multirow{2}{*}{Dataset} & \multicolumn{2}{c||}{\texttt{K2-TREE}} & \multicolumn{2}{c||}{\texttt{K2-PDF}} & \multicolumn{2}{c||}{\texttt{K2-CPDF (P32)}}  & \multicolumn{2}{c||}{\texttt{K2-CPDF (PW)}}&  \multicolumn{2}{c||}{\texttt{K2-EDF}} &\multicolumn{2}{c||}{\texttt{K2-BP}} & \multicolumn{2}{c|}{\texttt{K2-CBP}}\\
%    \cline{2-15}
%     & bits & min & bits & min & bits & min & bits & min & bits & min & bits & min & bits & min \\
%     \hline
%     IN-2004 & 3.12 & 1.07 & 2.93 & 2.20 & 2.06 & min & 1.97 & min & 2.97 & 1.26 & 7.80 & 2.61 & 2.90 & 4.40\\ 
%     \hline
%     ID-2004 & 2.56 & 125 & 2.41 & 54 & 1.34 & min & 1.28 & min & 2.42 & 34 & 6.20 & 124 & 1.87 & 164\\ 
%     \hline
%     AR-2005 & 2.92 & 85 & 2.75 & 291 & 1.73 & min & 1.67 & min & 2.75 & 134 & 7.24 & 254 & 2.37 & 335\\ 
%     \hline
%     UK-2005 & 2.89 & 89 & 2.72 & 901 & 1.69 & min & 1.64 & min & 2.73 & 429 & 7.16 & 425 & 2.35 & 909\\ 
%     \hline
%    \end{tabular}
%    \end{center}
%    \caption{Results for matrix multiplication for web graph matrices; table reports the compressed matrix size in bits per nonzero and running time in minutes.}
%    \label{tab:webgraph_result}
%\end{table}

\begin{figure}
    \begin{center}
    \begin{tikzpicture}[scale=0.7]
     \begin{axis}[
            xlabel={Density},
            ylabel={Average time in second},
            title={Multiplication},
            log ticks with fixed point,
            x dir=reverse,
            xmode=log,
            xtick={1e-1, 1e-2, 1e-3, 1e-4},
            legend pos=south west
        ]
            %nav
            \addplot[sharp plot,mark=o, mark size=4pt,red] plot coordinates {
                (2e-1,8.89626301) (1e-1, 4.11452052) (1e-2, 0.18019336) (1e-3, 0.01881509) (1e-4, 0.00502304)
            };

            %pdft
            \addplot[sharp plot,mark=+, mark size=4pt,blue] plot coordinates {
                (2e-1, 3.92592779) (1e-1, 2.94502825) (1e-2, 0.31178314) (1e-3, 0.03191211) (1e-4, 0.00646658)
            };

            %edf
            \addplot[sharp plot, mark=triangle, mark size = 4pt, blue] plot coordinates {
                (2e-1, 3.80699486) (1e-1, 2.83599754) (1e-2, 0.29608632) (1e-3, 0.03044228) (1e-4, 0.00608669)
            };

            %bp
            \addplot[sharp plot, mark=+, mark size = 4pt] plot coordinates {
                (2e-1, 15.60183152) (1e-1, 10.62632997) (1e-2, 0.78554955) (1e-3, 0.05015245) (1e-4, 0.00963482)
            };

            %cbp
            \addplot[sharp plot, mark=star, mark size = 4pt] plot coordinates {
                (2e-1, 22.03675870) (1e-1, 14.59138857) (1e-2, 2.16418309) (1e-3, 0.09817016) (1e-4, 0.01170682)
            };
        \end{axis}
    \end{tikzpicture}
    \end{center}
    \begin{center}
    \begin{tikzpicture}[scale=0.7]
        \begin{axis}[%
            hide axis, xmin=10, xmax=50, ymin=0, ymax=0.4,
            legend style={draw=white!15!black,legend cell align=left}, legend columns=3,
            ]
            \addlegendimage{mark=o, mark size=4pt,color = red}
            \addlegendentry{\texttt{K2-TREE}};
            \addlegendimage{mark=+, mark size=4pt,color = blue}
            \addlegendentry{\texttt{K2-PDFT}};
            \addlegendimage{mark=triangle, mark size = 4pt, color = blue}
            \addlegendentry{\texttt{K2-EDF}};
            \addlegendimage{mark=+, mark size = 4pt, color=black}
            \addlegendentry{\texttt{K2-BP}};
            \addlegendimage{mark=star, mark size = 4pt, color=black}
            \addlegendentry{\texttt{K2-CBP}};
        \end{axis}
    \end{tikzpicture}
    \end{center}
    \caption{Average computation time for multiplying randomly generated uniform-distributed matrices as a function of their density}
    \label{plot:random_mult}
\end{figure}
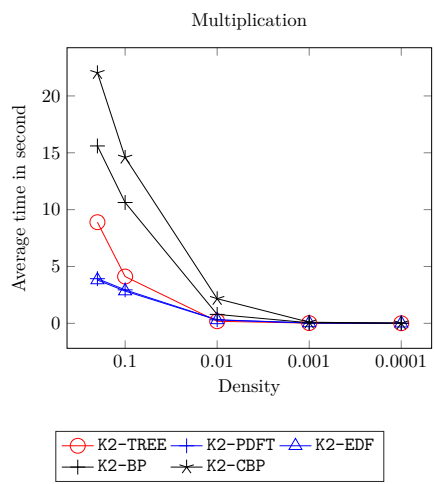

For matrix-matrix multiplication on randomly generated matrices the results in Figure~\ref{plot:random_mult} indicate that \texttt{K2-PDF} and \texttt{K2-EDF} outperform the classical \texttt{K2-TREE}, particularly when the matrix density is higher. In contrast, when the matrices are more sparse, all methods exhibit similar performance, as the multiplication time remains low due to the relatively small size of the matrices tested.
This behavior can be explained by the impact of density on the tree structure. A slight increase in density results in larger subtrees, requiring the traversal of a greater number of them. In such scenarios, \texttt{K2-PDF} and \texttt{K2-EDF} benefit from their simpler representations, which allow them to read matrices sequentially from left to right. Conversely, the \texttt{K2-TREE} must perform an additional \texttt{rank} operation to determine the starting position of the four submatrices, introducing an overhead that affects performance.

% need to write this space values
% https://docs.google.com/spreadsheets/d/1FZu7GY-rFTl7Oc-KiFrblbXTcaa_pnQQl9gkqtQ3A78/edit?usp=sharing
\begin{table}
    \begin{center}
    \begin{tabular}{|l||r|r||r|r||r|r||r|r||r|r|}
    \hline
       \multirow{2}{*}{Dens.} & \multicolumn{2}{c||}{\texttt{K2-TREE}} &  \multicolumn{2}{c||}{\texttt{K2-PDF}} &  \multicolumn{2}{c||}{\texttt{K2-EDF}} &\multicolumn{2}{c||}{\texttt{K2-BP}} & \multicolumn{2}{c|}{\texttt{K2-CBP}} \\
    \cline{2-11}
     & bits & sec & bits & sec & bits & sec & bits & sec & bits& sec\\
     \hline
     0.2       & 4.86 & 8.90 & 4.59 & 3.93 & 4.91 & 3.81 & 13.86 & 15.60 & 11.98 & 22.04\\ 
     \hline
     0.1       & 6.73 & 4.11 & 6.32 & 2.95 & 6.95 & 2.84 & 19.37 & 10.63 & 17.70 & 14.59\\ 
     \hline
     $10^{-2}$ & 13.58 & 0.18 & 12.61 & 0.31 & 14.25 & 0.30 & 41.98 & 0.79 & 23.45 & 2.16\\ 
     \hline
     $10^{-3}$ & 22.14 & 0.02 & 19.25 & 0.03 & 30.80 & 0.03 & 55.19 & 0.05 & 43.29 & 0.10\\ 
     \hline
     $10^{-4}$ & 44.03 & 0.01 & 25.85 & 0.01 & 58.17 & 0.01 & 109.02 & 0.01 & 149.14 & 0.01\\ 
     \hline
    \end{tabular}        
    \end{center}
    \caption{Results for matrix multiplication for random matrices; table reports average total space by the data structure in bits per nonzero and average running time in seconds.}
    \label{tab:random_mult}
\end{table}
\ignore{
\begin{table}
    \begin{center}
    \begin{tabular}{|l||r|r||r|r||r|r||r|r||r|r|}
    \hline
       \multirow{2}{*}{Dens.} & \multicolumn{2}{c||}{\texttt{K2-TREE}} &  \multicolumn{2}{c||}{\texttt{K2-PDF}} &  \multicolumn{2}{c||}{\texttt{K2-EDF}} &\multicolumn{2}{c||}{\texttt{K2-BP}} & \multicolumn{2}{c|}{\texttt{K2-CBP}} \\
    \cline{2-11}
     & bits & sec & bits & sec & bits & sec & bits & sec & bits& sec\\
     \hline
      $10^{-1}$ & 6.71 & 4085.90& 6.31 & & 6.35 & 1445.50 & 17.44 & & 8.39 & \\ 
     \hline
     $10^{-2}$ & 13.39 & 81.1 & 12.59 & 146.80 & 12.98 & 142.73 & 34.68 & 374.82 & 19.67 & 1392.93\\ 
     \hline
     $10^{-3}$ & 20.42 & 2.32 & 19.18 & 6.83 & 20.15 & 6.52 & 54.48 & 14.10 & 32.35 & 50.16\\ 
     \hline
     $10^{-4}$ & 27.67 & 0.08 & 25.84 & 0.28 & 28.42 & 0.25 & 77.01 & 0.48 & 44.96 & 1.01\\ 
     \hline 
    \end{tabular}        
    \end{center}
    \caption{Results for matrix multiplication for random matrices; table reports average total space by the data structure in bits per nonzero and average running time in seconds.}
    \label{tab:random_mult2}
\end{table}}

In terms of compression, considering  that $k^2$-trees were introduced to exploit the clustering of nonzeros, it is interesting that 
even when the nonzero entries are uniformly distributed the results in  Table~\ref{tab:random_mult} show that there is some compression. We observe that \texttt{K2-PDF} uses fewer bits per nonzero element compared to \texttt{K2-TREE}, with the difference becoming more pronounced as the matrix density decreases. This is because \texttt{K2-PDF} does not require auxiliary data structures for its operations. This effect is further reinforced when examining the space usage of \texttt{K2-EDF}, which consumes more bits per nonzero element than both \texttt{K2-TREE} and \texttt{K2-PDF}. 
As expected, the representations \texttt{K2-BP} and \texttt{K2-CBP} require more space due to the BP representation and the additional data structures needed for navigation. However, a surprising observation is that even with uniformly distributed values, \texttt{K2-CBP} consistently requires fewer bits than \texttt{K2-BP} across all densities except for matrices with a density of $10^{-4}$. At such a low density, the number of nonzero elements is minimal, making it unlikely to find identical subtrees within a matrix of this size.

\section{Conclusions and further works}\label{sec:conc}

In this paper we have introduced new representations of $k^2$-trees in which nodes are encoded using a depth-first layout. The rationale is that such layout is more cache friendly and can be used to detect and exploit the presence of identical subtrees. Our experimental results on trees representing the adjacency matrices of web graphs have shown that repeated subtrees are indeed present and our Compressed Balanced Parenthesis representation is the one obtaining the best compression, even if at the cost of higher running times for the matrix-matrix multiplication. Our results also show that if the input matrix is not too dense our Enriched Depth First representation takes less space and is faster than the canonical $k^2$-tree representation. 
Although still preliminary, our results show that depth first layouts can bring advantages for some classes of matrices and operations, and we plan to pursue this idea along the following lines:
\begin{itemize}
    \item we plan to implement the search of identical subtrees also for the Enriched Depth First representation: we notice that instead of working at the single bit level is it possible to do the computation considering the input as an array of blocks of $k^2$ bits;
    \item we plan to improve the performance of our balanced parenthesis representations by writing our implementation of the \findc\ operation: at the moment we are relying on the sdsl-lite library which is not space efficient since it supports also other BP operations;
    \item we plan to experiment with the DFUDS representation with possibly some of the improvements from~\cite{jansson2012};
    \item because of its recursive nature, matrix multiplcation can be easily parallelized: we plan to extend our comparison to a multithread setting;
    \item we are currently exploiting the presence of repeated subtrees only to save space; since they corresponds to identical submatrices  
    we plan to investigate whether they can lead to a speed-up for some matrix operations.
\end{itemize}

\bibliography{lipics-v2021-sample-article}

\end{document}